\documentclass{LMCS}

\def\doi{8 (2:12) 2012}
\lmcsheading%
{\doi}
{1--21}
{}
{}
{Jan.~11, 2011}
{Jun.~20, 2012}
{}

\usepackage{hareng}

\hypersetup{%
pdftitle={Untyping Typed Algebras and Colouring Cyclic Linear Logic},%
pdfauthor={Damien Pous},%
pdfkeywords={residuated lattices, cyclic linear logic, Kleene algebra,
  decision procedures, sequent calculus, typed algebra, proof search},%
}

\begin{document}

\title[Untyping Typed Algebras and Colouring Cyclic Linear Logic]%
{Untyping Typed Algebras \\ and Colouring Cyclic Linear Logic\rsuper*}%
\author[D.~Pous]{Damien Pous}%
\address{CNRS (LIG, UMR 5217, Grenoble)}%
\email{Damien.Pous@ens-lyon.fr}
\thanks{Partially funded by the French projects ``Choco'',
  ANR-07-BLAN-0324 and ``PiCoq'', ANR-10-BLAN-0305.}%
\titlecomment{{\lsuper*}Extended version of the abstract that appeared in
  Proc. CSL'10~\cite{pous:csl10:utas}.}%





\begin{abstract}
  \noindent We prove ``untyping'' theorems: in some typed theories
  (semi\-rings, Kleene algebras, residuated lattices, involutive
  residuated lattices), typed equations can be derived from the
  underlying untyped equations. As a consequence, the corresponding
  untyped decision procedures can be extended for free to the typed
  settings.
  Some of these theorems are obtained via a detour through fragments
  of cyclic linear logic, and give rise to a substantial optimisation
  of standard proof search algorithms.
\end{abstract}

\maketitle

\section*{Introduction}

\subsection*{Motivations.} 

The literature contains many decidability or complexity results for
various algebraic structures.  Some of these structures (rings, Kleene
algebras~\cite{Koz94b}, residuated lattices~\cite{OnoK85}) can be
generalised to \emph{typed} structures, where the elements come with a
domain and a codomain, and where operations are defined only when
these domains and codomains agree according to some simple
rules. Although such typed structures are frequently encountered in
practice (e.g., rectangular matrices, heterogeneous binary relations,
or more generally, categories), there are apparently no proper tools
to easily reason about these.

This is notably problematic in proof assistants, where powerful
decision procedures are required to let the user focus on difficult
reasoning steps by leaving administrative details to the
computer. Indeed, although some important theories can be decided
automatically in Coq or HOL (e.g., Presburger
arithmetic~\cite{Norrish03}, elementary real
algebra~\cite{Harrison93}, rings~\cite{GregoireM05}), there are no
high-level tools to reason about heterogeneous relations or
rectangular matrices.

\medskip

In this paper, we show how to extend the standard decision procedures
from the untyped structures to the corresponding typed structures. In
particular, we make it possible to use standard tools to reason about
rectangular matrices or heterogeneous relations, without bothering
about types (i.e., matrix dimensions or domain/codomain information).
The approach we propose is depicted below: we study ``untyping''
theorems that allow one to prove typed equations as follows: 1) erase
type informations, 2) prove the equation using standard, untyped,
decision procedures, and 3) derive a typed proof from the untyped one.
\begin{mathpar}
  \xymatrix @C=.9em {
    {\txt{untyped setting:}\quad}&
    &{\hat a}\ar@{=}^{\txt{\small{decide}}}[rr] &&{\hat b}&
    \ar|-{\txt{\small{rebuild types}}}[d]\\
    {\txt{typed setting:}}\quad&
    \ar|-{\txt{\small{erase types}}}[u]
    &{a}\ar@{=}[r]& {?}\ar@{=}[r]&{b}&
  }
\end{mathpar}
Besides the theoretical aspects, an important motivation behind this
work comes from a Coq library~\cite{atbr:itp} in which we
developed efficient tactics for partial axiomatisations of relations:
the ideas presented here were used and integrated in this library to
extend our tactics to typed structures, for free.

\subsection*{Overview.}

We shall mainly focus on the two algebraic structures we mentioned
above, since they raise different problems and illustrate several
aspects of these untyping theorems: Kleene algebras~\cite{Kle56} and
residuated lattices~\cite{Jipsen:Survey}.
\begin{iteMize}{$\bullet$}
\item The case of Kleene algebras is the simplest one. The main
  difficulty comes from the annihilating element $(0)$: its
  polymorphic typing rule requires us to show that equational proofs
  can be factorised so as to use the annihilation laws at first, and
  then reason using the other axioms.

\item The case of residuated structures is more involved: due to the
  particular form of axioms about residuals, we cannot rely on
  standard equational axiomatisations of these structures. Instead, we
  need to exploit an equivalent cut-free sequent proof system (first
  proposed by Ono and Komori~\cite{OnoK85}), and to notice that this
  proof system corresponds to the intuitionistic fragment of cyclic
  linear logic~\cite{Yetter90}. The latter logic is much more concise
  and the corresponding proof nets are easier to reason about, so that
  we obtain the untyping theorem in this setting. We finally port the
  result back to residuated lattices by standard means.
\end{iteMize}
The above sequent proof systems have the sub-formula property, so that
they yield decision procedures, using proof search algorithms. As an
unexpected application, we show that the untyping theorem makes it
possible to improve these algorithms by reducing the set of proofs
that have to be explored.

\subsection*{Outline.}

We introduce our notations and make the notion of typed structure
precise in~§\ref{sec:defs}. We study Kleene algebras and residuated
lattices in~§\ref{sec:ka} and~§\ref{sec:rl}, respectively. The
optimisation of proof search is analysed in~§\ref{sec:optim}; we
conclude with related work, and directions for future work
in~§\ref{sec:ccl}.

\section{Notation, typed structures}
\label{sec:defs}



Let $\mathcal X$ be an arbitrary set of \emph{variables}, ranged over
using letters $x,y$. Given a signature $\Sigma$, we let $a,b,c$ range
over the set $T(\Sigma+\mathcal X)$ of \emph{terms with variables}.
Given a set $\mathcal T$ of \emph{objects} (ranged over using letters
$n,m,p,q$), a \emph{type} is a pair $(n,m)$ of objects (which we
denote by $\h n m$, following categorical notation), a \emph{type
  environment} $\Gamma: \mathcal X\to\mathcal T^2$ is a function from
variables to types, and we will define \emph{type judgements} of the
form $\typ[\Gamma] a n m$, to be read ``in environment $\Gamma$, term
$a$ has type $\h n m$, or, equivalently, $a$ is a morphism from $n$ to
$m$''. By $\typ[\Gamma] {a,b} n m$, we mean that both $a$ and $b$ have
type $\h n m$; type judgements will include the following rule for
variables: %
\begin{mathpar} 
  \inferrule*[Right=Tv]{\Gamma(x)=(n,m)}{\typ[\Gamma] x n m} \and
\end{mathpar}
Similarly, we will define \emph{typed equality} judgements of the form
$\teq[\Gamma] a b n m$: ``in environment $\Gamma$, terms $a$ and $b$
are equal, at type $\h n m$''. Equality judgements will generally
include the following rules, so as to obtain an equivalence relation
at each type: %
\begin{mathpar}
  \inferrule*[Right=v]{\Gamma(x)=(n,m)}{\teq[\Gamma] x x n m} \and
  \inferrule*[Right=t]
  {\teq[\Gamma] a b n m \\\\ \teq[\Gamma] b c n m}
  {\teq[\Gamma] a c n m} \and 
  \inferrule*[Right=s]
  {\teq[\Gamma] a b n m}{\teq[\Gamma] b a n m}
\end{mathpar}
%
%
By taking the singleton set as set of objects $(\mathcal
T=\set\emptyset)$, we recover standard, untyped structures: the only
typing environment is $\Hamma: x\mapsto(\emptyset,\emptyset)$, and
types become uninformative (this corresponds to working in a
one-object category; all operations are total functions). To alleviate
notations, since the typing environment will always be either $\Hamma$
or an abstract constant value $\Gamma$, we shall leave it implicit in
type and equality judgements, by relying on the absence or the
presence of types to indicate which one to use. For example, we shall
write $\teq a b n m$ for $\teq[\Gamma] a b n m$, while $\peq a b$ will
denote the judgement $\teq[\Hamma] a b\emptyset\emptyset$.

\medskip

The question we study in this paper is the following one: given a
signature and a set of inference rules defining a type judgement and
an equality judgement, does the implication below hold, for all
$a,b,n,m$? 
\begin{align*}
  \begin{cases}
    \typ {a,b} n m\\
    \peq a b
  \end{cases}
  \quad\text{ entails }\quad \teq a b n m\enspace.
\end{align*}
In other words, in order to prove an equality in a typed structure, is
it safe to remove all type annotations, so as to work in the untyped
underlying structure?

\section{Kleene algebras}
\label{sec:ka}

\noindent 
We study the case of residuated lattices in~§\ref{sec:rl}; here
we focus on Kleene algebras.  In order to illustrate our methodology,
we actually give the proof in three steps, by considering two
intermediate algebraic structures: monoids and 
semirings. The former admit a rather simple and direct proof, while
the latter are sufficient to expose concisely the main difficulty in
handling Kleene algebras.

\clearpage 
\subsection{Monoids}
\label{ss:monoids} 

\begin{defi}
  \label{def:monoid}
  \emph{Typed monoids} are defined by the signature
  $\{\cdot{}_2,1_0\}$, together with the following inference rules, in
  addition to the rules from~§\ref{sec:defs}.
\begin{mathpar}
  \inferrule*[Right=To]{ }{\typ 1 n n} \and
  \inferrule*[Right=Td]{\typ a n m \and \typ b m p}{\typ {a\cdot{}b} n p} \and
  \inferrule*[Right=o]{ }{\teq 1 1 n n} \and
  \inferrule*[Right=d]{\teq a {a'} n m \and \teq b {b'} m p }
  {\teq {a\cdot{}b} {a'\cdot{}b'} n p} \and
  \inferrule*[Right=od]{\typ a n m}{\teq {1\cdot{}a} a n m} \and
  \inferrule*[Right=da]{\typ a n m \and \typ b m p \and \typ c p q}
  {\teq {(a\cdot{}b)\cdot{}c} {a\cdot{}(b\cdot{}c)} n q} \and
  \inferrule*[Right=do]{\typ a n m}{\teq {a\cdot{}1} a n m}
  \end{mathpar}
\end{defi}

\noindent 
In other words, typed monoids are just categories: $1$ and $\cdot$
correspond to identities and composition.  Rules \rul o and \rul d
ensure that equality is reflexive at each type (point $(i)$ below) and
preserved by composition. As expected, equalities relate correctly
typed terms only $(ii)$:
\begin{lem}\label{lem:m:sanity}
  \begin{enumerate}[(i)]
  \item If $\typ a n m$, then $\teq a a n m$.
  \item If $\teq a b n m$, then $\typ {a,b} n m$.
 \end{enumerate} 
\end{lem}

\noindent Moreover, in this setting, type judgements enjoy some form
of injectivity (types are not uniquely determined due to the unit
$(1)$, which is typed in a polymorphic way): 

\begin{lem}
  \label{lem:m:tinj}
  If $\typ a n m$ and $\typ a {n'} {m'}$, then we have $n=n'$ iff $m=m'$.
\end{lem}

\noindent We need another lemma to obtain the untyping theorem:
all terms related by the untyped equality admit the same type derivations.

\begin{lem}
  \label{lem:m:teqtyp}
  If $\peq a b$; then for all $n,m$, we have $\typ a n m$ iff $\typ b n m$.
\end{lem}

\begin{thm}
  \label{thm:m:untype}
  If $\peq a b$  and $\typ {a,b} n m$, then $\teq a b n m$.
\end{thm}
\begin{proof}
  We reason by induction on the derivation $\peq a b$; the
  interesting cases are the following ones:
  \begin{iteMize}{$\bullet$}
  \item the last rule used is the transitivity rule \rul t: we have
    $\peq a b$, $\peq b c$, $\typ {a,c} n m$, and we need to show that
    $\teq a c n m$. By Lemma~\ref{lem:m:teqtyp}, we have $\typ b n m$,
    so that by the induction hypotheses, we get $\teq a b n m$ and
    $\teq b c n m$, and we can apply rule \rul t.
  \item the last rule used is the compatibility of $\cdot{}$ \rul d:
    we have $\peq a a'$, $\peq b b'$, $\typ {a\cdot{}b,a'\cdot{}b'} n
    m$, and we need to show that $\teq {a\cdot{}b} {a'\cdot{}b'} n
    m$. By case analysis on the typing judgements, we deduce that
    $\typ a n p$, $\typ b p m$, $\typ {a'} n q$, $\typ {b'} q m$, for
    some $p,q$.  Thanks to Lemmas~\ref{lem:m:tinj}
    and~\ref{lem:m:teqtyp}, we have $p=q$, so that we can conclude
    using the induction hypotheses ($\teq a {a'} n p$ and $\teq b {b'}
    p m$), and rule \rul d. \qedhere
  \end{iteMize}
\end{proof}

\noindent Note that the converse of Theorem~\ref{thm:m:untype} ($\teq
a b n m$ entails $\peq a b$) is straightforward, so that we actually
have an equivalence.

\subsection{Non-commutative semirings}
\label{sec:semiring}

\begin{defi}
  \label{def:semiring}
  \emph{Typed semirings} are defined by the
  signature $\{\cdot{}_2,+_2,1_0,0_0\}$, together with the following
  rules, in addition to the rules from Def.~\ref{def:monoid}
  and~§\ref{sec:defs}.
  \begin{mathpar}
    \inferrule*[Right=Tz]{ }{\typ 0 n m} \and
    \inferrule*[Right=Tp]{\typ {a,b} n m}{\typ {a+b} n m}\and
    \inferrule*[Right=p]{\teq a {a'} n m \and \teq b {b'} n m}
    {\teq {a+b} {a'+b'} n m} \and
    \inferrule*[Right=z]{ }{\teq 0 0 n m} \and
    \inferrule*[Right=pz]{\typ a n m}{\teq {a+0} a n m} \and
    \inferrule*[Right=pc]{\typ {a,b} n m}{\teq {a+b} {b+a} n m} \and
    \inferrule*[Right=pa]{\typ {a,b,c} n m}{\teq {(a+b)+c} {a+(b+c)} n m} \and
    \inferrule*[Right=dp]{\typ a n m \and \typ{b,c} m p}
    {\teq {a\cdot{}(b+c)} {a\cdot{}b+a\cdot{}c} n p} \and
    \inferrule*[Right=dz]{\typ a n m}{\teq {a\cdot{}0} 0 n p} \and
    \inferrule*[Right=zd]{\typ a n m}{\teq {0\cdot{}a} 0 p m} \and
    \inferrule*[Right=pd]{\typ a n m \and \typ{b,c} p n}
    {\teq {(b+c)\cdot{}a} {b\cdot{}a+c\cdot{}a} p m} 
 \end{mathpar}
\end{defi}
\noindent 
In other words, typed semiring are categories enriched over a
commutative monoid: each homset is equipped with a commutative monoid
structure (typing rules \rul{Tz,Tp} and rules \rul{p,pz,pc,pa}), 
composition distributes over these monoid structures (rules
\rul{dp,dz,pd,zd}).

\medskip

Lemma~\ref{lem:m:sanity} is also valid in this setting: equality is
reflexive and relates correctly typed terms only. However, due to the
presence of the annihilator element $(0)$, Lemmas~\ref{lem:m:tinj}
and~\ref{lem:m:teqtyp} no longer hold: $0$ has any type, and we have
$\peq {x\cdot{}0\cdot{}x} 0$ while $x\cdot{}0\cdot{}x$ only admits
$\Gamma(x)$ as a valid type. Moreover, some valid proofs cannot be
typed just by adding decorations: for example, $0=0\cdot{}a\cdot{}a=0$
is a valid untyped proof of $0=0$; however, this proof cannot be typed
if $a$ has a non-square type. Therefore, we have to adopt another
strategy: we reduce the problem to the annihilator-free case, by
showing that equality proofs can be factorised so as to use rules
\rul{pz}, \rul{dz}, and \rul{zd} at first, as oriented rewriting
rules.

\begin{defi}
  \label{def:s:clean}
  Let $a$ be a term; we denote by $\clean a$ the \emph{normal form of} $a$,
  obtained with the following convergent rewriting system:
  \begin{mathpar}
    a+0 \to a \and 
    0+a \to a \and 
    0\cdot{}a \to 0 \and 
    a\cdot{}0 \to 0 
  \end{mathpar}
  We say that $a$ is \emph{strict} if $\clean a\neq 0$. 
\end{defi}

This normalisation procedure preserves types and equality; moreover,
on strict terms, we recover the injectivity property of types we had
for monoids:

\begin{lem}
  \label{lem:s:eqclean}
  If $\typ a n m$, then $\typ {\clean a} n m$ and $\teq a {\clean a} n
  m$.
\end{lem}

\begin{lem}
  \label{lem:s:tinj}
  For all strict terms $a$ such that $\typ a n m$ and $\typ a {n'}
  {m'}$, we have $n=n'$ iff $m=m'$.
\end{lem}

We can then define a notion of strict equality judgement, where the
annihilation laws are not allowed:
\begin{defi}
  \label{def:s:strict}
  We let $\steq[\_]\_\_\_\_$ denote the \emph{strict equality}
  judgement obtained by removing rules \rul{dz} and \rul{zd}, and
  replacing rules \rul{dp} and \rul{pd} with the following variants,
  where the factor has to be strict. %
  \begin{mathpar}
    \inferrule*[Right=dp$^+$]{\typ a n m \and \typ{b,c}m p \and\clean a\neq 0}
    {\steq {a\cdot{}(b+c)} {a\cdot{}b+a\cdot{}c} n p} \\
    \inferrule*[Right=pd$^+$]{\typ a n m \and\typ{b,c}p n \and\clean a\neq 0}
    {\steq {(b+c)\cdot{}a} {b\cdot{}a+c\cdot{}a} p m} 
  \end{mathpar}
\end{defi}

Using the same methodology as previously, one easily obtain the
untyping theorem for strict equality judgements.

\begin{lem}
  \label{lem:s:teqtyp}
  If $\speq a b$; then for all $n,m$, we have $\typ a n m$ iff $\typ b
  n m$.
\end{lem}

\begin{prop}
  \label{prp:s:suntype}
  If $\speq a b$  and $\typ {a,b} n m$, then $\steq a b n m$.
\end{prop}

\noindent Note that the patched rules for distributivity, \rul{dp$^+$}
and \rul{pd$^+$} are required in order to obtain
Lemma~\ref{lem:s:teqtyp}: if $a$ was not required to be strict, we
would have $\speq {0\cdot{}(x+y)} {0\cdot{}x + 0\cdot{}y}$, and the
right-hand side can be typed in environment $\Gamma=\{x\mapsto
(3,2),~y\mapsto (4,2)\}$ while the left-hand side cannot.

\medskip %
We now have to show that any equality proof can be factorised, so as
to obtain a strict equality proof relating the corresponding normal
forms:%
\begin{prop}
  \label{prp:s:factor}
  If $\peq a b$, then we have $\speq {\clean a} {\clean b}$.
\end{prop}
\begin{proof}
  We first show by induction that whenever $\peq a b$, $a$ is
  strict iff $b$ is strict $(\dagger)$. %
  Then we proceed by induction on the derivation $\peq a b$, we detail
  only some cases:
  \begin{enumerate}[\rul{dz}]
  \item[\rul d] we have $\speq{\clean a}{\clean {a'}}$ and
    $\speq{\clean b}{\clean {b'}}$ by induction; we need to show that
    $\speq{\clean {(a\cdot{}b)}}{\clean {(a'\cdot{}b')}}$. If one of
    $a,a',b,b'$ is not strict, then $\clean {(a\cdot{}b)}=\clean
    {(a'\cdot{}b')}=0$, thanks to $(\dagger)$, so that we are done;
    otherwise, $\clean {(a\cdot{}b)}=\clean a\cdot{}\clean b$, and
    $\clean {(a'\cdot{}b')}=\clean {a'}\cdot{}\clean {b'}$, so that we
    can apply rule \rul d.
  \item[\rul{dz}] trivial, since $\clean{(a\cdot{}0)}=0$.
  \item[\rul{dp}] we need to show that $\speq
    {\clean{(a\cdot{}(b+c))}} {\clean{(a\cdot{}b+a\cdot{}c)}}$; if one
    of $a,b,c$ is not strict, both sides reduce to the same term, so
    that we can apply Lemma~\ref{lem:m:sanity}$(i)$ (which holds in
    this setting); otherwise we have $\clean{(a\cdot{}(b+c))}=\clean
    a\cdot{}(\clean b+\clean c)$ and $\clean{(a\cdot{}b+a\cdot{}c)} =
    \clean a\cdot{}\clean b+\clean a\cdot{}\clean c$, so that we can
    apply rule \rul{dp$^+$}.\qedhere
  \end{enumerate}
\end{proof}

\smallskip\noindent
We finally obtain the untyping theorem by putting all together:

\begin{thm}
   \label{thm:s:untype}
   In semirings, for all $a,b,n,m$ such that $\typ {a,b} n m$, we
   have $\peq a b$ iff $\teq a b n m$.
\end{thm}
\begin{proof}
  The reverse implication is straightforward; we prove the direct one.
  By Lemma~\ref{lem:s:eqclean}, using the transitivity and symmetry
  rules, it suffices to show $\teq {\clean a} {\clean b} n m$. This is
  clearly the case whenever $\steq {\clean a} {\clean b} n m$, which
  follows from Props.~\ref{prp:s:factor} and~\ref{prp:s:suntype}.
\end{proof}

\subsection{Kleene algebras}~\medskip
\label{ss:ka} 

Kleene algebras are idempotent semirings equipped with a star
operation~\cite{Kle56}; they admit several important models, among
which binary relations and \emph{regular languages} (the latter is
complete~\cite{Krob91a,Koz94b}; since equality of regular languages is
decidable, so is the equational theory of Kleene algebras).  Like
previously, we type Kleene algebras in a natural way, where star
operates on ``square'' types: types of the form $n\to n$, i.e., square
matrices or homogeneous binary relations.

\begin{defi}
  \label{def:kleene}
  We define \emph{typed Kleene algebras} by the signature
  $\{\cdot{}_2,+_2,\star_1,1_0,$ $0_0\}$, together with the following
  rules, in addition that from Defs.~\ref{def:monoid}
  and~\ref{def:semiring}, and~§\ref{sec:defs}, and where $\tleq a b n
  m$ is an abbreviation for $\teq {a+b} b n m$.
  \begin{mathpar}
    \inferrule*[Right=Ts]{\typ a n n}{\typ {\tstar a} n n} \and
    \inferrule*[Right=s]{\teq a b n n}{\teq {\tstar a} {\tstar b} n n} \and
    \inferrule*[Right=pi]{\typ a n m}{\teq {a+a} a n m} \and
    \inferrule*[Right=sp]{\typ a n n}{\teq {1+a\cdot{}\tstar a} {\tstar a} n n} \and
    \inferrule*[Right=sl]{\tleq {a\cdot{}b} b n m}{\tleq {\tstar a\cdot{}b} b n m} \and
    \inferrule*[Right=sr]{\tleq {b\cdot{}a} b n m}{\tleq {b\cdot{}\tstar a} b n m} 
  \end{mathpar}
\end{defi}
\noindent 
The untyped version of this axiomatisation is that from
Kozen~\cite{Koz94b}: axiom \rul{pi} corresponds to idempotence of $+$,
the three other rules define the star operation (we omitted the mirror
image of axiom \rul{sp}, which is derivable from the other
ones~\cite{atbr:itp}).  Note that due to rules~\rul{sl} and
\rul{sr}, we are no longer in a purely equational setting; indeed, the
algebra of regular events is not finitely based~\cite{redko64}.

\medskip

The proof of the untyping theorem for Kleene algebras is obtained
along the lines of the proof for non-commutative semirings. We just
highlight the main differences here, complete proofs are available as
Coq scripts~\cite{this:web}.
First, it is a simple exercise to check that the following lemma
holds:
\begin{lem}
  For all $n$, we have $\teq {\tstar 0} 1 n n$.
\end{lem}

\noindent This allows us to extend the rewriting system from
Def.~\ref{def:s:clean} with the rule $\tstar 0 \to 1$, so that the
annihilator can also be removed in this setting. In particular, we 
obtain:
\begin{lem}
  \label{lem:k:eqclean}
  If $\typ a n m$, then $\typ {\clean a} n m$ and $\teq a {\clean a} n
  m$.
\end{lem}
\begin{lem}
  \label{lem:k:tinj}
  For all strict terms $a$ such that $\typ a n m$ and $\typ a {n'}
  {m'}$, we have $n=n'$ iff $m=m'$.
\end{lem}

\noindent Let $\steq[\_]\_\_\_\_$ denote the \emph{strict equality}
judgement obtained like previously (Def.~\ref{def:s:strict}), and
where we moreover adapt rules \rul{sl} and \rul{sr} so that $b$ is
required to be strict:
\begin{mathpar}
  \inferrule*[Right=sl$^+$]{\stleq {a\cdot{}b} b n m \and \clean b\neq 0}%
                           {\stleq {\tstar a\cdot{}b} b n m} \and
  \inferrule*[Right=sr$^+$]{\stleq {b\cdot{}a} b n m \and \clean b\neq 0}%
                           {\stleq {b\cdot{}\tstar a} b n m} 
\end{mathpar}

\noindent These patched rules \rul{sl$^+$} and \rul{sr$^+$} are
required to obtain the following counterpart to
Lemma~\ref{lem:s:teqtyp}: otherwise, we would have $\spleq {\tstar
  a\cdot{}0} 0$, where the right-hand side has any type while the type
of the left-hand side is constrained by $a$.

\begin{lem}
  \label{lem:k:teqtyp}
  If $\speq a b$; then for all $n,m$, we have $\typ a n m$ iff $\typ b
  n m$.
\end{lem}
\begin{proof}
  Similar to the proof of Lemma~\ref{lem:s:teqtyp}. Recall that
  $\spleq a b$ is an abbreviation for $\speq {a+b} b$; the rule
  \rul{sl$^+$} is handled as follows. Suppose that $\speq {\tstar
    a\cdot{}b+b} b$ was obtained using this rule:
  \begin{iteMize}{$\bullet$}
  \item if $\typ{\tstar a\cdot{}b+b} n m$, then we necessarily have
    $\typ b n m$;
  \item conversely, if $\typ b n m$ then we have $\typ{a\cdot{}b+b} n
    m$ by induction. Therefore, there exists $p$ such that $\typ a n
    p$ and $\typ b p m$. Since $b$ was required to be strict, we can
    use Lemma~\ref{lem:k:tinj} to deduce $n=p$, $\typ a n n$, and
    finally, $\typ{\tstar a\cdot{}b+b} n m$.
  \end{iteMize}
  Rule \rul{sr$^+$} is handled symmetrically, and rule \rul{sp} is
  straightforward.
\end{proof}

\noindent The untyping theorem for strict equality follows easily:
\begin{prop}
  \label{prp:k:suntype}
  If $\speq a b$  and $\typ {a,b} n m$, then $\steq a b n m$.
\end{prop}
\begin{proof}
  Like for Theorem~\ref{thm:m:untype} and Prop.~\ref{prp:s:suntype},
  we proceed by induction on the untyped derivation to add type
  annotations. We detail the case of rule \rul{sl$^+$}: suppose
  that $\spleq {\tstar a\cdot{}b} b$ was obtained using the untyped
  version of rule \rul{sl$^+$}, and $\typ {\tstar a\cdot{}b+b,b} n
  m$. Necessarily, $\typ a n n$ and $\typ {a\cdot{}b+b} n m$, so that we
  have $\tleq {a\cdot{}b} b n m$ by induction. We conclude using the
  typed version of rule \rul{sl$^+$}: $\tleq {\tstar a\cdot{}b} b n
  m$.
\end{proof}

\noindent We finally have to prove that Kleene algebra equality proofs
can be factorised using the strict equality judgement:
\begin{prop}
  \label{prp:k:factor}
  If $\peq a b$, then we have $\speq {\clean a} {\clean b}$.
\end{prop}
\begin{proof}
  By induction on the derivation, like for Prop.~\ref{prp:s:factor}.
  We detail only the rules involving Kleene star:
  \begin{iteMize}{$\bullet$}
  \item \rul{sp}: if $\clean a=0$ then $\clean{(1+a\cdot{}\tstar
      a)}=\clean{(\tstar a)}=1$ so that we can apply \rul o;
    otherwise, $\clean{(1+a\cdot{}\tstar a)}=1+\clean a\cdot{}\tstar
    {\clean a}$ and $\clean{(\tstar a)}=\tstar {\clean a}$: we can
    apply \rul{sp}.
  \item \rul{sl}: suppose that $\pleq{\tstar a\cdot{}b} b$ was
    obtained using this rule, we have to show that $\spleq{\clean
      {(\tstar a\cdot{}b)}} {\clean b}$. If $\clean b=0$ then $\clean
    {(\tstar a\cdot{}b)}=0$ and we use rule \rul z. Otherwise $b$ is
    strict, and either $\clean a=0$, in which case $\clean{(\tstar
      a\cdot{}b)}=1\cdot\clean b$, and we can use rules \rul{od} and
    \rul{pi} to get $\spleq{1\cdot\clean b}{\clean b}$; or $a$ is also
    strict. In the latter case, we use the induction hypothesis:
    $\spleq{\clean{(a\cdot{}b)}} {\clean b}$, i.e., $\spleq{\clean
      a\cdot\clean b} {\clean b}$, and we conclude using rule
    \rul{sl$^+$}.
  \item \rul{sr}: symmetric to the previous case. 
  \end{iteMize}
  (Note that we implicitly use the fact that normalisation commutes
  with sum, so that we have $\spleq {\clean a}{\clean b}$ iff $\speq
  {\clean{(a+b)}} {\clean b}$.)
\end{proof}

\begin{thm}
  \label{thm:k:untype}
  In Kleene algebras, for all $a,b,n,m$ such that $\typ {a,b} n m$, we
  have $\peq a b$ iff $\teq a b n m$.
\end{thm}

\subsection{Non-commutative rings}~\medskip
\label{ss:ring}

Before moving to residuated structures, we briefly discuss the case of
non-commutative rings. Indeed, although rings are quite similar to
semirings, they cannot be handled in the same way.

\begin{defi}
  \label{def:ring}
  We define \emph{typed rings} by the signature
  $\{\cdot{}_2,+_2,-_1,1_0,$ $0_0\}$, together with the following
  rules, in addition that from Defs.~\ref{def:monoid}
  and~\ref{def:semiring}, and~§\ref{sec:defs}.
  \begin{mathpar}
    \inferrule*[Right=Ti]{\typ a n m}{\typ {-a} n m} \and
    \inferrule*[Right=i]{\teq a b n m}{\teq {-a} {-b} n m} \and
    \inferrule*[Right=pi]{\typ a n m}{\teq {a+(-a)} 0 n m} 
  \end{mathpar}
\end{defi}

Due to the axiom \rul{pi}, we cannot define a simple function to
remove annihilators and obtain a factorisation system. Indeed, we have
$\peq a b$ iff $\peq {a+(-b)} 0$, so that strictness amounts to
provability; we no longer have a simple syntactical criterion.
However, unlike terms of Kleene algebras, terms of non-commutative
rings can easily be put in normal form (by expanding the underlying
polynomials and ordering monomials lexicographically---assuming that the set
of variables is ordered). This allows us to obtain the untyping
theorem by reasoning about the normalisation function.

Let $\nf a$ denote the normal form of the term $a$ (we do not define
formally this standard function here since we are mainly interested in
the methodology).

\begin{prop}
  For all $a,b,n,m$, we have
  \begin{enumerate}[\em(i)]
  \item $\peq a b$ iff $\nf a = \nf b$;
  \item if $\typ a n m$, then $\typ {\nf a} n m$;
  \item if $\typ a n m$, then $\teq {\nf a} a n m$.
  \end{enumerate}
\end{prop}
\begin{proof}\hfill
  \begin{enumerate}[(i)]
  \item Standard: this is the correctness and completeness of the
    untyped decision procedure: two expressions are equal if and only
    if they share the same normal form.
  \item By a straightforward induction on the typing derivation.
  \item Also by induction on the typing derivation, it amounts to
    replaying the standard correctness proof and checking that it is
    actually well-typed. \qedhere
  \end{enumerate}
\end{proof}

\noindent 
The untyping theorem follows immediately:
\begin{cor}
  In non-commutative rings, for all $a,b,n,m$ such that $\typ {a,b} n
  m$, we have $\peq a b$ iff $\teq a b n m$.
\end{cor}
\begin{proof}
  If $\peq a b$ then $\nf a = \nf b$ by the point $(i)$ above, which
  entails $\teq {\nf a} {\nf b} n m$ by reflexivity since $\typ {\nf
    a} n m$ by $(ii)$, from which we deduce $\teq a b n m$ by
  $(iii)$. The converse implication is straightforward, as in the
  previous sections.
\end{proof}

\section{Residuated lattices}
\label{sec:rl}

\noindent
We now move to our second example, \emph{residuated lattices}. These
structures also admit binary relations as models; they are of special
interest to reason algebraically about well-founded relations. For
example, residuation is used to prove Newman's Lemma in relation
algebras~\cite{DBW97}.  We start with a simpler structure.

\medskip

A \emph{residuated monoid} is a tuple
$(X,\leq,\cdot{},1,\backslash,/)$, such that $(X,\leq)$ is a partial
order, $(X,\cdot{},1)$ is a monoid whose product is monotonic ($a\leq
a'$ and $b\leq b'$ entail $a\cdot{}b \leq a'\cdot{}b'$), and
$\backslash,/$ are binary operations, respectively called \emph{left}
and \emph{right divisions}, characterised by the following
equivalences: 
\begin{align*} a\cdot{}b \leq c
  \quad\Leftrightarrow\quad b \leq a\backslash c
  \quad\Leftrightarrow\quad a \leq c/b
\end{align*}

\noindent 
Such a structure can be typed in a natural way, by using the following
rules for left and right divisions:
\begin{mathpar}
  \inferrule*[Right=Tl]{\typ c n m \and \typ a n p}%
  {\typ {a\backslash c} p m} \and
  \inferrule*[Right=Tr]{\typ c n m \and \typ b p m}%
  {\typ {c/b} n p}
\end{mathpar}

\noindent
Although we can easily define a set of axioms to capture equalities
provable in residuated monoids~\cite{Jipsen:Survey}, the transitivity
rule $\rul T$ becomes problematic in this setting (there is no
counterpart to Lemma~\ref{lem:m:teqtyp}). Instead, we exploit a
characterisation due to Ono and Komori~\cite{OnoK85}, based on a
Gentzen proof system for the full Lambek
calculus~\cite{Lambek58}. Indeed, the ``cut'' rule corresponding to
this system, which plays the role of the transitivity rule, can be
eliminated. Therefore, this characterisation allows us to avoid the
problems we encountered with standard equational proof systems. In
some sense, moving to cut-free proofs corresponds to using a
factorisation system, like we did in the previous section
(Prop.~\ref{prp:s:factor}).

\subsection{Gentzen proof system for residuated monoids}~\medskip
\label{ss:gpsrm}

Let $l,k,h$ range over lists of terms, let $l;k$ denote the
concatenation of $l$ and $k$, and let $\epsilon$ be the empty
list. The Gentzen proof system is presented on Fig.~\ref{fig:uimll};
it relates lists of terms to terms. It is quite
standard~\cite{Jipsen:Survey}: there is an axiom rule~\rul V, and, for
each operator, an introduction and an elimination rule.
\begin{figure}[tb]
  \centering 
  \begin{mathpar}
    \inferrule*[Right=v ]{ }{\pseq x x} \and
    \inferrule*[Right=Io]{ }{\pseq \epsilon 1} \and
    \inferrule*[Right=Id]{\pseq l a \and \pseq {l'}{a'}}{\pseq {l;l'} {a\cdot{}a'}} \and
    \inferrule*[Right=Ir]{\pseq {l;b} a}{\pseq l {a/b}} \and
    \inferrule*[Right=Il]{\pseq {b;l} a}{\pseq l {b\backslash a}} \\
    \inferrule*[Right=Eo]{\pseq {l;l'} a}{\pseq {l;1;l'} a} \and
    \inferrule*[Right=Ed]{\pseq {l;b;c;l'} a}{\pseq {l;b\cdot{}c;l'} a} \and
    \inferrule*[Right=Er]{\pseq k b \and \pseq {l;c;l'} a}{\pseq {l;c/b;k;l'} a} \and
    \inferrule*[Right=El]{\pseq k b \and \pseq {l;c;l'} a}{\pseq {l;k;b\backslash c;l'} a}
  \end{mathpar}
  \caption{Gentzen proof system for residuated monoids.}
  \label{fig:uimll}
\end{figure}
The axiom rule can be generalised to terms~$(i)$, the cut rule is
admissible~$(ii)$, and the proof system is correct and complete
w.r.t.\ residuated monoids~$(iii)$.
\begin{prop}
 \begin{enumerate}[(i)]
  \item For all $a$, we have $\pseq a a$.
  \item For all $l,k,k',a,b$ such that $\pseq l a$ and
    $\pseq {k;a;k'} b$, we have $\pseq {k;l;k'} b$.
  \item For all $a,b$, we have $\pseq a b$ iff $a\leq b$ holds in all
    residuated monoids.
  \end{enumerate}
\end{prop}
\begin{proof}
  Point $(i)$ is easy; see~\cite{OnoK85,OkadaT99,Jipsen:Survey} for cut
  admissibility and completeness.
\end{proof}

\noindent
Type decorations can be added to the proof system in a straightforward
way (see Fig.~\ref{fig:imll}). However, using this proof system, we
were able to prove the untyping theorem only for the unit-free
fragment: we needed to assume that terms have at most one type, which
is not true in the presence of $1$. This proof was rather involved, so
that we did not manage to circumvent this difficulty in a nice and
direct way. Instead, as hinted in the introduction, we move to the
following more symmetrical setting.

\begin{figure}[tb]
  \centering 
  \begin{mathpar}
    \inferrule*[Right=v ]{\Gamma(x)=(n,m)}{\seq x x n m} \and
    \inferrule*[Right=Io]{ }{\seq \epsilon 1 n n} \and
    \inferrule*[Right=Eo]{\seq {l;l'} a n m}{\seq {l;1;l'} a n m} \and
    \inferrule*[Right=Id]{\seq l a n m \and \seq {l'}{a'} m p}%
                         {\seq {l;l'} {a\cdot{}a'} n p} \and
    \inferrule*[Right=Ed]{\seq {l;b;c;l'} a n m}%
                         {\seq {l;b\cdot{}c;l'} a n m} \and
    \inferrule*[Right=Ir]{\typ b p m \and \seq {l;b} a n m}%
                         {\seq l {a/b} n p} \and
    \inferrule*[Right=Er]{\typ {l'} m q \and \seq k b n m \and \seq {l;c;l'} a p q}%
                         {\seq {l;c/b;k;l'} a p q} \and
    \inferrule*[Right=Il]{\typ b n p \and \seq {b;l} a n m}%
                         {\seq l {b\backslash a} p m} \and
    \inferrule*[Right=El]{\typ l p m \and \seq k b m n \and \seq {l;c;l'} a p q}%
                         {\seq {l;k;b\backslash c;l'} a p q} 
  \end{mathpar}
  \caption{Typed Gentzen proof system for residuated monoids.}
  \label{fig:imll}
\end{figure}

\subsection{Cyclic MLL}~\medskip
\label{ss:cymll} 

The sequent proof system for residuated monoids (Fig.~\ref{fig:uimll})
actually corresponds to a non-commutative version of intuitionistic
multiplicative linear logic (IMLL)~\cite{Girard87:ll}: the product
$(\cdot{})$ is a non-commutative tensor~$(\otimes)$, and left and
right divisions $(\backslash,/)$ are the corresponding left and right
linear implications ($\multimap,\multimapinv$).  Moreover, it happens
that this system is just the intuitionistic fragment of cyclic
multiplicative linear logic (MLL)~\cite{Yetter90}.  The untyping
theorem turned out to be easier to prove in this setting, which we
describe below.

We assume a copy $\mathcal X^\bot$ of the set of variables $(\mathcal
X)$, and we denote by $x^\bot$ the corresponding elements which we
call \emph{dual variables}. From now on, we shall consider terms with
both kinds of variables: $T(\Sigma+X+X^\bot)$.  We keep an algebraic
terminology to remain consistent with the previous sections; notice
that using terminology from logic, a term is a formula and a variable
is an atomic formula.

\begin{defi}
  \label{def:MLL:terms}
  \emph{Typed MLL terms} are defined by the signature
  $\{\otimes_2,\parr_2,1_0,\bot_0\}$, together with the following
  typing rules:
  \begin{mathpar}
    \inferrule*[Right=Tv]{\Gamma(x)=(n,m)}{\typ x n m} \and
    \inferrule*[Right=T$_1$]{ }{\typ 1 n n} \and
    \inferrule*[Right=T$_\otimes$]{\typ a n m\and\typ b m p}{\typ{a\otimes b} n p} \\
    \inferrule*[Right=Tv$^\bot$]{\Gamma(x)=(n,m)}{\typ {x^\bot} m n}\and
    \inferrule*[Right=T$_\bot$]{ }{\typ \bot n n} \and
    \inferrule*[Right=T$_\parr$]{\typ a n m \and \typ b m p}{\typ {a\parr b} n p} 
  \end{mathpar}
\end{defi}
\noindent
Tensor $(\otimes)$ and par $(\parr)$ are typed like the previous dot
operation; bottom $(\bot)$ is typed like the unit $(1)$; dual
variables are typed by mirroring the types of the corresponding
variables. We extend type judgements to lists of terms as follows:
\begin{mathpar}
  \inferrule*[Right=Te]{ }{\typ \epsilon n n} \and
  \inferrule*[Right=Tc]{\typ a n m \and \typ l m p}{\typ {a;l} n p}
\end{mathpar} 
(be careful not to confuse $\typ {a,b} n m$, which indicates that both
$a$ and $b$ have type $\h n m$, with $\typ{a;b} n m$, which indicates
that the list $a;b$ has type $\h n m$).
\emph{Linear negation} is defined over terms and lists of terms as follows:
\begin{align*}
  (x)^\bot & \eqdef x^\bot &
  1^\bot & \eqdef \bot & 
  (a\otimes b)^\bot & \eqdef b^\bot\parr a^\bot & 
  (a;l)^\bot & \eqdef l^\bot;a^\bot \\
  (x^\bot)^\bot & \eqdef x&
 \bot^\bot & \eqdef 1 & 
  (a\parr b)^\bot & \eqdef b^\bot\otimes a^\bot & 
 \epsilon^\bot & \eqdef \epsilon
\end{align*} 
Note that since we are in a non-commutative setting, negation has to
reverse the arguments of tensors and pars, as well as lists. Negation
is involutive and mirrors type judgements:
\begin{lem}
  \label{lem:lneg:facts}
  For all $l$, $l^{\bot\bot}=l$; for all $l,n,m$, $\typ l n m$ iff
  $\typ {l^\bot} m n$.
\end{lem}

\noindent 
If we were using a two-sided presentation of MLL, judgements would be
of the form $\seq l k m n$, intuitively meaning ``$\pseq l k$ is
derivable in cyclic MLL, and lists $l$ and $k$ have type $\h m
n$''. Instead, we work with one-sided sequents to benefit from the
symmetrical nature of MLL. At the untyped level, this means that we
replace $\pseq l k$ with $\pceq {l^\bot;k}$. According to the previous
intuitions, the list $l^\bot;k$ has a square type $\h n n$: the object
$m$ is hidden in the concatenation, so that it suffices to record the
outer object~$(n)$. Judgements finally take the form $\ceq n l$,
meaning ``the one-sided MLL sequent $\vdash l$ is derivable at type
$\h n n$''.

\begin{defi}
  \label{def:MLL}
  \emph{Typed cyclic MLL} is defined by the sequent calculus from
  Fig.~\ref{fig:typed-cMLL}.
\end{defi}
\begin{figure}[tb]
  \centering 
  \begin{mathpar}
    \inferrule*[Right=1]{ }{\ceq n 1} \and %
    \inferrule*[Right=$\bot$]{\ceq n l}{\ceq n {\bot;l}} \and %
    \inferrule*[Right=$\otimes$]{\ceq n {l;a}\and \ceq n {b;k}}{\ceq n {l;a\otimes b;k}} \and %
    \inferrule*[Right=$\parr$]{\ceq n {a;b;l}}{\ceq n {a\parr b;l}} \\ %
    \inferrule*[Right=A]{\Gamma(x)=(n,m)}{\ceq m {x^\bot;x}} \and %
    \inferrule*[Right=E]{\typ a n m \and \ceq m {l;a}}{\ceq n {a;l}}
  \end{mathpar}
  \caption{Typed proof system for Cyclic MLL.}
  \label{fig:typed-cMLL}
\end{figure}

\noindent
Except for type decorations, the system is standard: the five first
rules are the logical rules of MLL~\cite{Girard87:ll}. Rule \rul E is
the only structural rule, this is a restricted form of the exchange
rule, yielding cyclic permutations: sequents have to be thought of as
rings~\cite{Yetter90}.  As before, we added type decorations in a
minimal way, so as to ensure that derivable sequents have square
types, as explained above:
\begin{lem}
  For all $l,n$, if $\ceq n l$ then $\typ l n n$.
\end{lem}

We now give a graphical interpretation of the untyping theorem, using
proof nets.  Since provability is preserved by cyclic permutations,
one can draw proof structures by putting the terms of a sequent on a
circle~\cite{Yetter90}. For example, a proof $\pi$ of a sequent
$\pceq{l_0,\dots,l_i}$ will be represented by a proof net whose
interface is given by the left drawing below.

\smallskip
\begin{center}
 \begin{picture}(0,0)%
\includegraphics{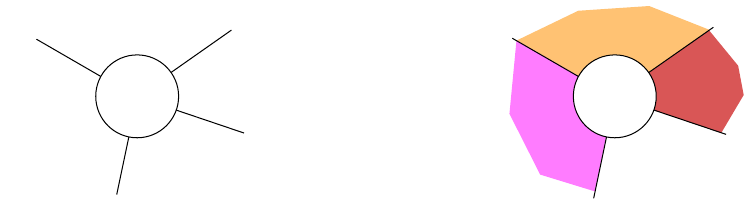}%
\end{picture}%
\setlength{\unitlength}{2565sp}%
\begingroup\makeatletter\ifx\SetFigFont\undefined%
\gdef\SetFigFont#1#2#3#4#5{%
  \reset@font\fontsize{#1}{#2pt}%
  \fontfamily{#3}\fontseries{#4}\fontshape{#5}%
  \selectfont}%
\fi\endgroup%
\begin{picture}(5493,1636)(2401,-3246)
\put(2416,-1936){\makebox(0,0)[lb]{\smash{{\SetFigFont{8}{9.6}{\familydefault}{\mddefault}{\updefault}{\color[rgb]{0,0,0}$l_0$}%
}}}}
\put(4137,-1793){\makebox(0,0)[lb]{\smash{{\SetFigFont{8}{9.6}{\familydefault}{\mddefault}{\updefault}{\color[rgb]{0,0,0}$l_1$}%
}}}}
\put(4226,-2653){\makebox(0,0)[lb]{\smash{{\SetFigFont{8}{9.6}{\familydefault}{\mddefault}{\updefault}{\color[rgb]{0,0,0}$l_2$}%
}}}}
\put(3083,-3168){\makebox(0,0)[lb]{\smash{{\SetFigFont{8}{9.6}{\familydefault}{\mddefault}{\updefault}{\color[rgb]{0,0,0}$l_i$}%
}}}}
\put(3809,-2870){\rotatebox{30.0}{\makebox(0,0)[b]{\smash{{\SetFigFont{8}{9.6}{\familydefault}{\mddefault}{\updefault}{\color[rgb]{0,0,0}$\dots$}%
}}}}}
\put(5982,-1936){\makebox(0,0)[lb]{\smash{{\SetFigFont{8}{9.6}{\familydefault}{\mddefault}{\updefault}{\color[rgb]{0,0,0}$l_0$}%
}}}}
\put(6610,-3168){\makebox(0,0)[lb]{\smash{{\SetFigFont{8}{9.6}{\familydefault}{\mddefault}{\updefault}{\color[rgb]{0,0,0}$l_i$}%
}}}}
\put(7336,-2870){\rotatebox{30.0}{\makebox(0,0)[b]{\smash{{\SetFigFont{8}{9.6}{\familydefault}{\mddefault}{\updefault}{\color[rgb]{0,0,0}$\dots$}%
}}}}}
\put(7796,-2653){\makebox(0,0)[lb]{\smash{{\SetFigFont{8}{9.6}{\familydefault}{\mddefault}{\updefault}{\color[rgb]{0,0,0}$l_2$}%
}}}}
\put(7713,-1793){\makebox(0,0)[lb]{\smash{{\SetFigFont{8}{9.6}{\familydefault}{\mddefault}{\updefault}{\color[rgb]{0,0,0}$l_1$}%
}}}}
\put(7440,-2271){\makebox(0,0)[lb]{\smash{{\SetFigFont{8}{9.6}{\familydefault}{\mddefault}{\updefault}{\color[rgb]{0,0,0}$n_2$}%
}}}}
\put(6321,-2566){\makebox(0,0)[lb]{\smash{{\SetFigFont{8}{9.6}{\familydefault}{\mddefault}{\updefault}{\color[rgb]{0,0,0}$n_0$}%
}}}}
\put(6852,-1916){\makebox(0,0)[lb]{\smash{{\SetFigFont{8}{9.6}{\familydefault}{\mddefault}{\updefault}{\color[rgb]{0,0,0}$n_1$}%
}}}}
\put(3427,-2368){\makebox(0,0)[b]{\smash{{\SetFigFont{8}{9.6}{\familydefault}{\mddefault}{\updefault}{\color[rgb]{0,0,0}$\pi$}%
}}}}
\put(6954,-2368){\makebox(0,0)[b]{\smash{{\SetFigFont{8}{9.6}{\familydefault}{\mddefault}{\updefault}{\color[rgb]{0,0,0}$\pi$}%
}}}}
\end{picture}%

\end{center}

\noindent 
Suppose now that the corresponding list admits a square type: $\typ l
n n$, i.e., $\forall j\leq i, \typ{l_j}{n_j}{n_{j+1}}$, for some
$n_0,\dots,n_{i+1}$ with $n=n_0=n_{i+1}$. One can add these type
decorations as background colours, in the areas delimited by terms, as
we did on the right-hand side.

The logical rules of the proof system (Fig.~\ref{fig:typed-cMLL}) can
then be represented by the proof net constructions from
Fig.~\ref{fig:proofnetsrules} (thanks to this sequent representation,
the exchange rule \rul E is implicit).  Since these constructions
preserve planarity, all proof nets are planar~\cite{BellinF98}, and
the idea of background colours makes sense. Moreover, they can be
coloured in a consistent way, so that typed derivations correspond to
proof nets that can be entirely and consistently coloured.
\begin{figure}[t]
  \centering
  \begin{center}
    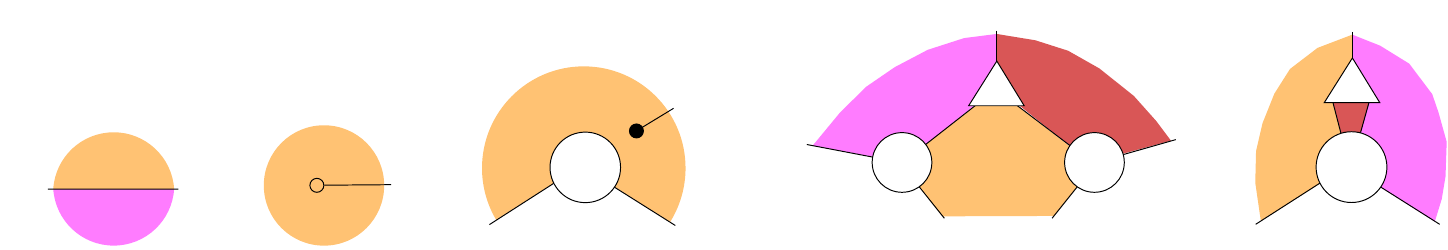
  \end{center}
  \caption{Proof nets for Cyclic MLL.}
  \label{fig:proofnetsrules}
\end{figure}
Therefore, one way to prove the untyping theorem consists in showing
that any proof net whose outer interface can be coloured can be
coloured entirely.
As an example, we give an untyped derivation below, together with the
corresponding proof net. Assuming that $\Gamma(x)=\h n m$ and
$\Gamma(y)=\h m p$, the conclusion has type $\h p p$, and the outer
interface of the proof net can be coloured (here, with colours $p$ and
$n$). The untyping theorem will ensure that there exists a typed
proof; indeed, the whole proof net can be coloured in a consistent
way.

\smallskip
\begin{center}
  \qquad\qquad
  \begin{minipage}{0.55\linewidth}
   $
      \inferrule*[Right=\scriptsize$\parr$]{
        \inferrule*[Right=\scriptsize E]{
          \inferrule*[Right=\scriptsize $\bot$]{
            \inferrule*[Right=\scriptsize$\otimes$]{
              \inferrule*[Right=\scriptsize$\parr$]{
                \inferrule*[Right=\scriptsize$\otimes$]{
                  \inferrule*[Right=\scriptsize  A]{
                  }{\pceq{x^\bot;x}} 
                  \and
                  \inferrule*[Right={\scriptsize E,A}]{ 
                  }{\pceq{y; y^\bot}} 
                }{\pceq{x^\bot;(x\otimes y); y^\bot}} 
              }{\pceq{x^\bot;(x\otimes y)\parr y^\bot}}
              \and
              \inferrule*[Right={\scriptsize E,A}]{ 
              }{\pceq{y;y^\bot}}
            }{\pceq{x^\bot;((x\otimes y)\parr y^\bot)\otimes y;y^\bot}}
          }{\pceq{\bot;x^\bot;((x\otimes y)\parr y^\bot)\otimes y;y^\bot}}
        }{\pceq{y^\bot ;\bot; x^\bot;((x\otimes y)\parr y^\bot)\otimes y}}
      }{\pceq{y^\bot\parr\bot\parr x^\bot;((x\otimes y)\parr y^\bot)\otimes y}}
    $
 \end{minipage}
 \hfill
 \begin{minipage}{0.32\linewidth}
   \begin{picture}(0,0)%
\includegraphics{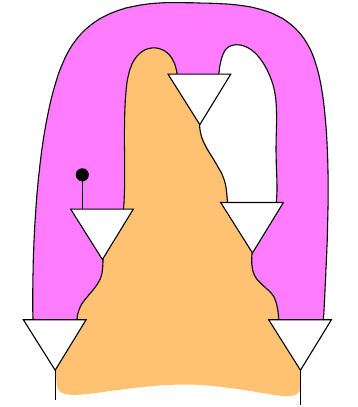}%
\end{picture}%
\setlength{\unitlength}{2763sp}%
\begingroup\makeatletter\ifx\SetFigFont\undefined%
\gdef\SetFigFont#1#2#3#4#5{%
  \reset@font\fontsize{#1}{#2pt}%
  \fontfamily{#3}\fontseries{#4}\fontshape{#5}%
  \selectfont}%
\fi\endgroup%
\begin{picture}(2441,2780)(3317,-861)
\put(4591, 19){\makebox(0,0)[b]{\smash{{\SetFigFont{8}{9.6}{\familydefault}{\mddefault}{\updefault}{\color[rgb]{0,0,0}$n$}%
}}}}
\put(4021,307){\makebox(0,0)[b]{\smash{{\SetFigFont{8}{9.6}{\familydefault}{\mddefault}{\updefault}{\color[rgb]{0,0,0}$\parr$}%
}}}}
\put(4692,1256){\makebox(0,0)[b]{\smash{{\SetFigFont{8}{9.6}{\familydefault}{\mddefault}{\updefault}{\color[rgb]{0,0,0}$\otimes$}%
}}}}
\put(3697,-452){\makebox(0,0)[b]{\smash{{\SetFigFont{8}{9.6}{\familydefault}{\mddefault}{\updefault}{\color[rgb]{0,0,0}$\parr$}%
}}}}
\put(5050,352){\makebox(0,0)[b]{\smash{{\SetFigFont{8}{9.6}{\familydefault}{\mddefault}{\updefault}{\color[rgb]{0,0,0}$\parr$}%
}}}}
\put(5010,973){\makebox(0,0)[b]{\smash{{\SetFigFont{8}{9.6}{\familydefault}{\mddefault}{\updefault}{\color[rgb]{0,0,0}$p$}%
}}}}
\put(3332,435){\makebox(0,0)[b]{\smash{{\SetFigFont{8}{9.6}{\familydefault}{\mddefault}{\updefault}{\color[rgb]{0,0,0}$p$}%
}}}}
\put(5743,267){\makebox(0,0)[b]{\smash{{\SetFigFont{8}{9.6}{\familydefault}{\mddefault}{\updefault}{\color[rgb]{0,0,0}$p$}%
}}}}
\put(4651,1617){\makebox(0,0)[b]{\smash{{\SetFigFont{8}{9.6}{\familydefault}{\mddefault}{\updefault}{\color[rgb]{0,0,0}$m$}%
}}}}
\put(5382,-428){\makebox(0,0)[b]{\smash{{\SetFigFont{8}{9.6}{\familydefault}{\mddefault}{\updefault}{\color[rgb]{0,0,0}$\otimes$}%
}}}}
\end{picture}%

 \end{minipage}
\end{center}

\bigskip\noindent
We now embark in the proof of the untyping theorem for cyclic
MLL; the key property is that the types of derivable sequents are all
squares:
\begin{prop}
  \label{prop:ceq:mono}
  If $\pceq l$ and $\typ l n m$, then $n=m$.
\end{prop}
\begin{proof}
  We proceed by induction on the untyped derivation $\pceq l$, but we
  prove a stronger property: ``the potential types of all cyclic
  permutations of $l$ are squares'', i.e., for all $h$,$k$ such that
  $l=h;k$, for all $n,m$ such that $\typ{k;h} n m$,
  $n=m$. 
  The most involved case is that of the tensor rule. Using symmetry
  arguments, we can assume that the cutting point belongs to the left
  premise: the conclusion of the tensor rule is $\pceq l;l';a\otimes
  b;k$, we suppose that the induction hypothesis holds for $l;l';a$
  and $b;k$, and knowing that $\typ {l';a\otimes b;k;l} n m$, we have
  to show $n=m$. Clearly, we have $\typ {l';a} n p$, $\typ {b;k} p q$,
  and $\typ l q m$ for some $p,q$. By induction on the second premise,
  we have $p=q$, so that $\typ {l';a;l} n m$. Since the latter list is
  a cyclic permutation of $l;l';a$, we can conclude with the induction
  hypothesis on the first premise.
\end{proof}

\begin{thm}
  \label{thm:mll:untype}
  In cyclic MLL, if $\typ l n n$, then we have $\pceq l$ iff $\ceq n l$.
\end{thm}
\begin{proof}
  The right-to-left implication is straightforward; for the direct
  implication, we proceed by induction on the untyped derivation. The
  previous proposition is required in the case of the tensor rule: we
  know that $\pceq{l;a}$, $\pceq{b;k}$, and $\typ{l;a\otimes b;k} n
  n$, and we have to show that $\ceq n {l;a\otimes b;k}$. Necessarily,
  there is some $m$ such that $\typ{l;a} n m$ and $\typ{b;k} m n$;
  moreover, by Prop.~\ref{prop:ceq:mono}, $n=m$.  Therefore, we can
  apply the induction hypotheses (so that $\ceq n {l;a}$ and $\ceq n
  {b;k}$) and we conclude with the typed tensor rule.
\end{proof}

\subsection{Intuitionistic fragment}~\medskip
\label{ss:imll} 

To deduce that the untyping theorem holds in residuated monoids, it
suffices to show that the typed version of the proof system
from~§\ref{ss:gpsrm} corresponds to the intuitionistic fragment of the
proof system from Fig.~\ref{fig:typed-cMLL}. This is well-known for
the untyped case, and type decorations do not add particular
difficulties. Therefore, we just give a brief overview of the extended
proof.

The idea is to define the following families of \emph{input} and
\emph{output} terms (Danos-Regnier
polarities~\cite{regnierPhD,BellinS94}), and to work with sequents
composed of exactly one output term and an arbitrary number of input
terms.
\begin{align*}
  \begin{array}{r@{}c@{\OR}c@{\OR}c@{\OR}c@{\OR}c}
    i ::= & ~x^\bot & \bot & i \parr i & i\otimes o & o\otimes i \\[.3em]
    o ::= & x       & 1    & o \otimes o & i\parr o & o\parr i\\ 
  \end{array}
\end{align*}

\noindent
Negation $(-^\bot)$ establishes a bijection between input and output
terms. Terms of residuated monoids (IMLL formulae) are encoded into
output terms as follows.
\begin{align*}
  \enc {a\cdot b} &\eqdef \enc a \otimes \enc b &
  \enc {a/ b} &\eqdef \enc a \parr \enc b^\bot &
  \enc x &\eqdef x \\
  \enc 1 &\eqdef 1 &
  \enc {a\backslash b} &\eqdef \enc a^\bot \parr \enc b
\end{align*}

\noindent
This encoding is a bijection between IMLL terms and MLL output terms;
it preserves typing judgements:
\begin{lem}
  \label{lem:enc:typed}
  For all $a,n,m$, we have $\typ a n m$ iff $\typ {\enc a} n m$.
\end{lem}
\noindent 
(Note that we heavily rely on overloading to keep notation simple.)
The next proposition shows that we actually obtained a fragment of
typed cyclic MLL; it requires the lemma below: input-only lists are
not derivable. The untyping theorem for residuated monoids follows
using Thm.~\ref{thm:mll:untype}.
\begin{lem}
  \label{lem:pceq:ouput}
  If $\pceq l$, then $l$ contains at least one output term.
\end{lem}
\begin{prop}
  If $\typ {l,a} n m$, then $\seq l a n m$ iff $\ceq m {\enc
    l^\bot;\enc a}$.
\end{prop}
\begin{proof}
  The forward implication is proved by an induction on the sequent
  derivation. For the reverse direction, we actually prove the
  following stronger property, by induction on the untyped MLL
  derivation:
  \begin{quote}
    ``for all $h,a,k,n,m$ such that we have $\pceq {\enc h^\bot;\enc
      a;\enc k^\bot}$, $\typ {h;k} n m$, and $\typ a n m$, we have
    $\seq {h;k} a n m$''.
  \end{quote}
  This generalisation is required to handle the exchange rule.
  We detail only the key cases: %
  \begin{iteMize}{$\bullet$}
  \item If the tensor rule was used last, on the output term (which
    was thus of the form $\enc{a\cdot{}b}=\enc a\otimes \enc b$):
    \begin{align*}
      \inferrule*[Right=$\otimes$]{\pceq {\enc h^\bot;\enc a}\and
        \pceq {\enc b;\enc k^\bot}}{\pceq {\enc h^\bot;\enc a\otimes
          \enc b;\enc k^\bot}}
    \end{align*}
    Since $\typ {a\cdot b} n m$, and $\typ {h;k} n m$, we have $p,q$
    such that $\typ a n p$, $\typ b p m$, $\typ h n q$ and $\typ k q m$.
    Therefore, by Lemmas~\ref{lem:lneg:facts} and~\ref{lem:enc:typed},
    we have $\typ {\enc h^\bot;\enc a} q p$, whence $p=q$ by
    Prop.~\ref{prop:ceq:mono}. We can thus apply the induction
    hypothesis to the two premises to obtain $\seq h a n p$ and $\seq
    k b p m$ (using an empty sequence in both cases). We conclude
    using rule \rul{Id} from Fig.~\ref{fig:uimll}.

  \item If the tensor rule was used last, on one of the input terms,
    say on $b$ in $h=h_1;b;h_2$, with $\enc b^\bot=c\otimes d$:
    \begin{align*}
      \inferrule*[Right=$\otimes$] {\pceq {\enc {h_2}^\bot;c}\and
        \pceq {d;\enc {h_1}^\bot,\enc a,k^\bot}} {\pceq {\enc
          {h_2}^\bot;c\otimes d;\enc {h_1}^\bot;\enc a;\enc
          k^\bot}}
    \end{align*}
    Since $\enc h_2^\bot;c$ is provable and $\enc h_2^\bot$ contains
    only input terms, $c$ is necessarily an output term by
    Lemma~\ref{lem:pceq:ouput}. Therefore there is only one
    possibility ensuring $\enc b = d^\bot\parr c^\bot$: the term $b$ must
    be of the form $d'/c'$, with $\enc{d'}=d^\bot$ and $\enc{c'}=c$.

    We have $\typ {h_1;d'/c';h_2;k} n m$ and $\typ a n m$, i.e., $\typ
    {h_1} n p$, $\typ {d'} p q$, $\typ {c'} r q$, $\typ {h_2} r s$,
    and $\typ {k} s m$ for some $p,q,r,s$.  We first notice that the
    provable sequent ${\enc {h_2}^\bot;c}$ has type $\h s q$, so that
    $s=q$ by Prop.~\ref{prop:ceq:mono}. By induction, we then deduce
    $\seq {h_2} {c'} r q$ and $\seq {h_1;d';k} a n m$, and we conclude
    using rule \rul{Er} from Fig.~\ref{fig:uimll}. \qedhere
  \end{iteMize}
\end{proof}

\begin{cor}
  \label{cor:imll:untype}
  In residuated monoids, if $\typ {l,a} n m$, then we have $\pseq l a$
  iff $\seq l a n m$.
\end{cor}

\subsection{Residuated lattices: additives.}~\medskip
\label{ss:imall} 

The Gentzen proof system we presented for residuated monoids
(Fig.~\ref{fig:uimll}) was actually designed for residuated
\emph{lattices}~\cite{OnoK85}, obtained by further requiring the
partial order $(X,\leq)$ to be a lattice $(X,\vee,\wedge)$. Binary
relations fall into this family, by considering set-theoretic unions
and intersections.  The previous proofs scale without major
difficulty: on the logical side, this amounts to considering the
additive binary connectives $(\oplus,\with)$. By working in
multiplicative additive linear logic (MALL) without additive
constants, we get an untyping theorem for \emph{involutive residuated
  lattices}~\cite{Wille03:involrl}; we deduce the untyping theorem for
residuated lattices by considering the corresponding intuitionistic
fragment (see~\cite{this:web} for proofs).

%

\medskip

On the contrary, and rather surprisingly, the theorem breaks if we
include additive constants $(0,\top)$, or equivalently, if we consider
\emph{bounded} residuated lattices. The corresponding typing rules are
given below, together with the logical rule for top (there is no rule
for zero).
\begin{mathpar}
  \inferrule*[Right=T$_0$]{ }{\typ 0 n m} \and
  \inferrule*[Right=T$_\top$]{ }{\typ \top n m} \and
  \inferrule*[Right=$\top$]{\typ l m n}{\ceq n {\top;l}} 
\end{mathpar}
The sequent $x^\bot\otimes\top;y^\bot;\top\otimes x$ gives a
counter-example. This sequent basically admits the two following
untyped proofs:
\begin{mathpar}
  \inferrule*[Right=\scriptsize E]{ %
    \inferrule*[Right=\scriptsize$\otimes$]{ %
      \inferrule*[Right=\scriptsize{E,$\top$}]{ }{\pceq{y^\bot;\top}} %
      \and %
      \inferrule*[Right=\scriptsize$\otimes$]{ %
        \inferrule*[Right=\scriptsize{E,A}]{ }{\pceq{x;x^\bot}} %
        \and %
        \inferrule*[Right=\scriptsize$\top$]{ }{\pceq{\top}} %
      }{\pceq{x;x^\bot\otimes\top}} %
    }{\pceq{y^\bot;\top\otimes x;x^\bot\otimes\top}} %
  }{\pceq{x^\bot\otimes\top;y^\bot;\top\otimes x}} %
  \and %
  \inferrule*[Right=\scriptsize{E,E}]{ %
    \inferrule*[Right=\scriptsize$\otimes$]{ %
      \inferrule*[Right=\scriptsize$\top$]{ }{\pceq{\top}} %
      \and %
      \inferrule*[Right=\scriptsize$\otimes$]{ %
        \inferrule*[Right=\scriptsize{E,A}]{ }{\pceq{x;x^\bot}} %
        \and %
        \inferrule*[Right=\scriptsize $\top$]{ }{\pceq{\top;y^\bot}} %
      }{\pceq{x;x^\bot\otimes\top;y^\bot}} %
    }{\pceq{\top\otimes x;x^\bot\otimes\top;y^\bot}} %
 }{\pceq{x^\bot\otimes\top;y^\bot;\top\otimes x}} %
\end{mathpar}
However, this sequent admits the square type $\h m m$ whenever
$\Gamma(x)=(n,m)$ and $\Gamma(y)=(p,q)$, while the above proofs cannot
be typed unless $n=q$ or $n=p$, respectively. Graphically, these
proofs correspond to the proof nets below (where the proof net
construction for rule \rul{$\top$} is depicted on the left-hand side);
these proof nets cannot be coloured unless $n=q$ or $n=p$. 

\smallskip
\begin{center}
  \begin{picture}(0,0)%
\includegraphics{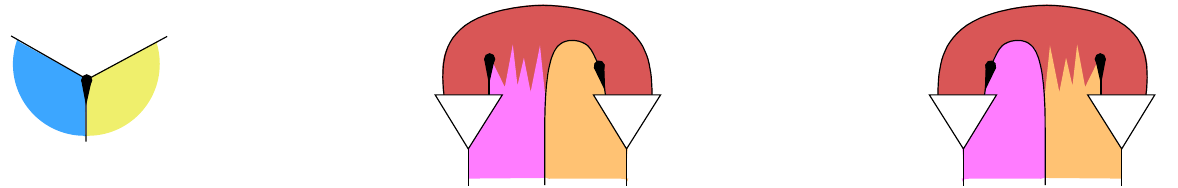}%
\end{picture}%
\setlength{\unitlength}{2960sp}%
\begingroup\makeatletter\ifx\SetFigFont\undefined%
\gdef\SetFigFont#1#2#3#4#5{%
  \reset@font\fontsize{#1}{#2pt}%
  \fontfamily{#3}\fontseries{#4}\fontshape{#5}%
  \selectfont}%
\fi\endgroup%
\begin{picture}(7516,1192)(772,-861)
\put(7955,-428){\makebox(0,0)[b]{\smash{{\SetFigFont{9}{10.8}{\familydefault}{\mddefault}{\updefault}{\color[rgb]{0,0,0}$\otimes$}%
}}}}
\put(1090,-298){\makebox(0,0)[b]{\smash{{\SetFigFont{9}{10.8}{\familydefault}{\mddefault}{\updefault}{\color[rgb]{0,0,0}$m$}%
}}}}
\put(1918,143){\makebox(0,0)[b]{\smash{{\SetFigFont{9}{10.8}{\familydefault}{\mddefault}{\updefault}{\color[rgb]{0,0,0}$l_i$}%
}}}}
\put(787,148){\makebox(0,0)[b]{\smash{{\SetFigFont{9}{10.8}{\familydefault}{\mddefault}{\updefault}{\color[rgb]{0,0,0}$l_1$}%
}}}}
\put(3765,-428){\makebox(0,0)[b]{\smash{{\SetFigFont{9}{10.8}{\familydefault}{\mddefault}{\updefault}{\color[rgb]{0,0,0}$\otimes$}%
}}}}
\put(3760,-191){\makebox(0,0)[b]{\smash{{\SetFigFont{9}{10.8}{\familydefault}{\mddefault}{\updefault}{\color[rgb]{0,0,0}$n$}%
}}}}
\put(4790,-191){\makebox(0,0)[b]{\smash{{\SetFigFont{9}{10.8}{\familydefault}{\mddefault}{\updefault}{\color[rgb]{0,0,0}$n$}%
}}}}
\put(7209,-733){\makebox(0,0)[b]{\smash{{\SetFigFont{9}{10.8}{\familydefault}{\mddefault}{\updefault}{\color[rgb]{0,0,0}$q$}%
}}}}
\put(7702,-733){\makebox(0,0)[b]{\smash{{\SetFigFont{9}{10.8}{\familydefault}{\mddefault}{\updefault}{\color[rgb]{0,0,0}$p$}%
}}}}
\put(6617,-733){\makebox(0,0)[b]{\smash{{\SetFigFont{9}{10.8}{\familydefault}{\mddefault}{\updefault}{\color[rgb]{0,0,0}$m$}%
}}}}
\put(8273,-733){\makebox(0,0)[b]{\smash{{\SetFigFont{9}{10.8}{\familydefault}{\mddefault}{\updefault}{\color[rgb]{0,0,0}$m$}%
}}}}
\put(7960,-191){\makebox(0,0)[b]{\smash{{\SetFigFont{9}{10.8}{\familydefault}{\mddefault}{\updefault}{\color[rgb]{0,0,0}$n$}%
}}}}
\put(6930,-191){\makebox(0,0)[b]{\smash{{\SetFigFont{9}{10.8}{\familydefault}{\mddefault}{\updefault}{\color[rgb]{0,0,0}$n$}%
}}}}
\put(4778,-428){\makebox(0,0)[b]{\smash{{\SetFigFont{9}{10.8}{\familydefault}{\mddefault}{\updefault}{\color[rgb]{0,0,0}$\otimes$}%
}}}}
\put(5103,-733){\makebox(0,0)[b]{\smash{{\SetFigFont{9}{10.8}{\familydefault}{\mddefault}{\updefault}{\color[rgb]{0,0,0}$m$}%
}}}}
\put(3447,-733){\makebox(0,0)[b]{\smash{{\SetFigFont{9}{10.8}{\familydefault}{\mddefault}{\updefault}{\color[rgb]{0,0,0}$m$}%
}}}}
\put(4510,-733){\makebox(0,0)[b]{\smash{{\SetFigFont{9}{10.8}{\familydefault}{\mddefault}{\updefault}{\color[rgb]{0,0,0}$p$}%
}}}}
\put(6942,-428){\makebox(0,0)[b]{\smash{{\SetFigFont{9}{10.8}{\familydefault}{\mddefault}{\updefault}{\color[rgb]{0,0,0}$\otimes$}%
}}}}
\put(4000,-733){\makebox(0,0)[b]{\smash{{\SetFigFont{9}{10.8}{\familydefault}{\mddefault}{\updefault}{\color[rgb]{0,0,0}$q$}%
}}}}
\put(1353, 31){\makebox(0,0)[b]{\smash{{\SetFigFont{9}{10.8}{\familydefault}{\mddefault}{\updefault}{\color[rgb]{0,0,0}$\dots$}%
}}}}
\put(1324,-748){\makebox(0,0)[b]{\smash{{\SetFigFont{9}{10.8}{\familydefault}{\mddefault}{\updefault}{\color[rgb]{0,0,0}$\top$}%
}}}}
\put(1531,-313){\makebox(0,0)[b]{\smash{{\SetFigFont{9}{10.8}{\familydefault}{\mddefault}{\updefault}{\color[rgb]{0,0,0}$n$}%
}}}}
\end{picture}%

\end{center}

\noindent
This counter-example for MALL also gives a counter-example for IMALL:
the above proofs translate to intuitionistic proofs of
$\pseq{y\cdot(\top\backslash x)}{\top\cdot x}$, which is also not
derivable in the typed setting, unless $n=q$ or $n=p$. 

The problem is actually even stronger: while $S\cdot(\top\backslash R)
\subseteq \top\cdot R$ holds for all homogeneous binary relations
$R,S$ (by the above untyped proofs, for example), this law does not
hold for arbitrary heterogeneous relations (see
Remark~\ref{rem:counter} below).
This shows that we cannot always reduce the analysis of typed
structures to that of the underlying untyped structures. Here, the
equational theory of heterogeneous binary relations does not reduce to
the equational theory of homogeneous binary relations.

\begin{rem}\label{rem:counter}
  The containment $S\cdot(\top\backslash R) \subseteq \top\cdot R$
  does not necessarily hold for all heterogeneous binary relations
  $R,S$, although it holds for all heterogeneous binary relations on
  non-empty sets.
\end{rem}
\begin{proof}
  Let $A,B,C,D$ be four sets, let $R\subseteq B\times C$ be a binary
  relation from $B$ to $C$, and let $S\subseteq D\times A$ be a binary
  relation from $D$ to $A$. To be precise, we denote by $\top_{X,Y}$
  the full relation between sets $X$ and $Y$ ($X\times Y$), and the
  containment from the statement can be rewritten as
  \begin{align*}
    S\cdot(\top_{B,A}\backslash R) \subseteq \top_{D,B}\cdot
    R\enspace.
  \end{align*}
  For all relations $T\subseteq B\times A$, the relation $T\backslash
  R$ is characterised as follows:
  \begin{align*}
    T\backslash R &= \set[\forall k\in B, (k,i)\in T \to (k,j)\in
    R]{(i,j)\in A\times C}\enspace.
  \end{align*}
  \begin{iteMize}{$\bullet$}
  \item if $B$ is the empty set, then $R=\top_{D,B}\cdot R=\emptyset$,
    and by the above characterisation, we have $\top_{B,A}\backslash
    R=A\times C$. Therefore, we can contradict the containment by
    taking any non-empty relation for $S$.
    (Note that this cannot happen in an homogeneous setting: we have
    $A=B=C=D$ so that taking the empty set for $B$ forces both $R$ and
    $S$ to be empty.)
  \item if $B$ is not empty, then we have 
    \begin{align*}
      \top_{B,A}\backslash R 
      &= \set[\forall k\in B, (k,i)\in\top_{B,A} \to (k,j)\in R]{(i,j)\in A\times C} \\
      &= \set[\forall k\in B, (k,j)\in R]{(i,j)\in A\times C} \\
      &\subseteq \set[\exists k\in B, (k,j)\in R]{(i,j)\in
        A\times C}\\
      &= \set[\exists k\in B, (i,k)\in \top_{A,B} \land (k,j)\in R]{(i,j)\in
        A\times C}\\
      &= \top_{A,B}\cdot R\enspace;
    \end{align*}
    Therefore, since $S\subseteq \top_{D,A}$, we can conclude: 
    \begin{align*}
      S\cdot(\top_{B,A}\backslash R) \subseteq
      \top_{D,A}\cdot\top_{A,B}\cdot R \subseteq \top_{D,B}\cdot
      R\enspace. 
    \end{align*}\qedhere
  \end{iteMize}
\end{proof}
\noindent 
We do not know whether relations on empty sets are required to get such
a counter-example in the model of binary relations. In other words,
for the signature of bounded residuated lattices, does the equational
theory of heterogeneous binary relations on non-empty sets reduce to
the equational theory of homogeneous binary relations?

\section{Improving proof search for residuated structures.}
\label{sec:optim}

The sequent proof systems we mentioned in the previous section have
the sub-formula property, so that provability is decidable in each
case, using a simple proof search
algorithm~\cite{OkadaT99}. Surprisingly, the concept of type can be
used to cut off useless branches. Indeed, recall
Prop.~\ref{prop:ceq:mono}: ``the types of any derivable sequent are
squares''. By contrapositive, given an untyped sequent $l$, one can
easily compute an abstract `most general type and environment' $(\h n
m,\Gamma)$, such that $\typ[\Gamma] l n m$ holds (taking $\mathbb N$
as the set of objects, for example); if $n\neq m$, then the sequent is
not derivable, and proof search can fail immediately on this sequent.

We did some experiments with a simple prototype~\cite{this:web}: we
implemented focused~\cite{Andreoli92:focalisation} proof search for
cyclic MALL, i.e., a recursive algorithm composed of an
\emph{asynchronous} phase which is deterministic and a
\emph{synchronous} phase, where branching occurs (e.g., when applying
the tensor rule \rul{$\otimes$}). The optimisation consists in
checking that the most general type of the sequent is square before
entering the synchronous phases. The overall complexity remains
exponential (provability is
NP-complete~\cite{Pentus06}---PSPACE-complete with
additives~\cite{kanovich65}) but we get an exponential speed-up: we
can abort proof search immediately on approximately two sequents out
of three.

The experimental results are given on Fig.~\ref{fig:bench}
and~\ref{fig:bench:distrib}---raw data is available
from~\cite{this:web}. We generate (pseudo) random sequents in normal
form with respect to the laws of neutral elements for multiplicative
constants ($1$ and $\bot$), with a given number of leaves (variables,
dual variables or constants), and where variables are picked in a set
of the specified size. E.g., $a\otimes \bot;b^\bot$ is a sequent with
three leaves and two variables, which can also be considered as a
sequent with three leaves and four variables, where two variables are
not used.

\begin{figure}[t]
  \noindent
  \begin{minipage}[t]{.49\linewidth}
    \includegraphics{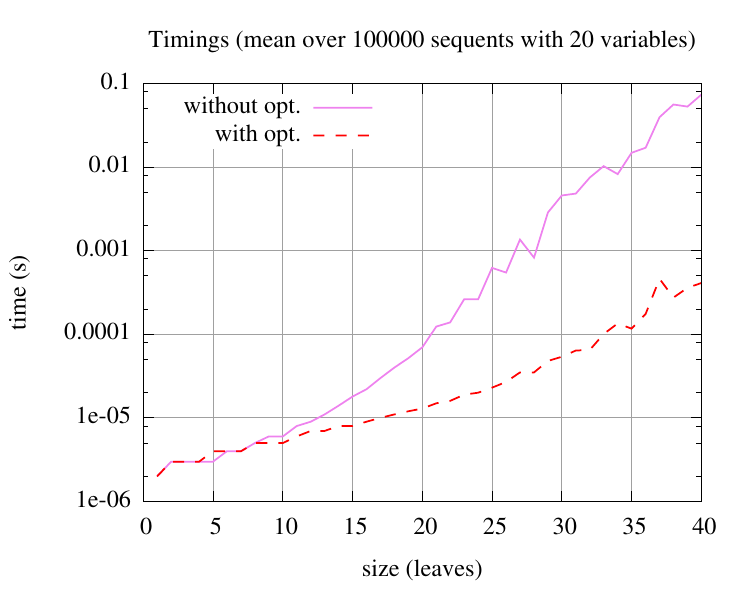}
  \end{minipage}
  \begin{minipage}[t]{.49\linewidth}
    \includegraphics{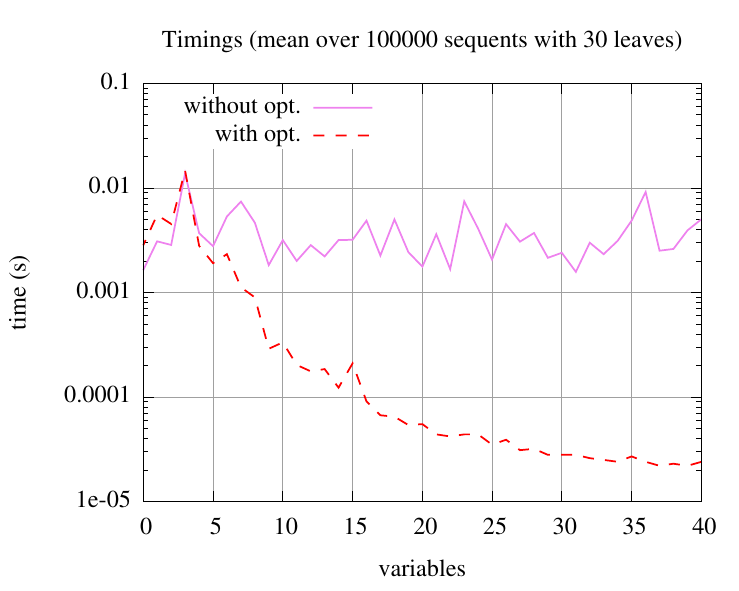}    
  \end{minipage}
  \caption{Searching times for focused proof search with and without
    optimisation.}
  \label{fig:bench}
\end{figure}

Each point of Fig.~\ref{fig:bench} was obtained by timing focused
proof search with and without optimisation, on a set of 100~000
sequents with the given characteristics: fixed number of variables and
varying size on the left-hand side, fixed size and varying number of
variables on the right-hand side. While the optimisation introduces
a small amount of overhead for very small sequents or sequents with few
variables, we gain more than one order of magnitude for larger
sequents. One can also notice that the more variables are available,
the more efficient the optimisation is: indeed, sequents with a lot of
different variables tend to have non-square types more easily, so that
they can be ruled out more frequently.

We did not report standard deviation in Fig.~\ref{fig:bench} since it
does not make sense in this setting: we have an unbounded set of
potential values, and the actual complexity of proof search is highly
stochastic. Instead, we computed the distribution of searching times:
Fig.~\ref{fig:bench:distrib} shows the proportion of sequents that are
solved in a given amount of time, among sequents with a fixed size and
number of variables---here, 30 leaves and 20 variables. While 60\% of
the sequents are solved in less that 10$^{-5}$s (with or without
optimisation), some of them require much more time: up to five minutes
without optimisation, and up to three seconds with the
optimisation. All in all, the overhead which is paid on `easily
solved' sequents gets compensated by the drastic improvement on
`harder' sequents.

\begin{figure}[t]
  \centering
  \includegraphics{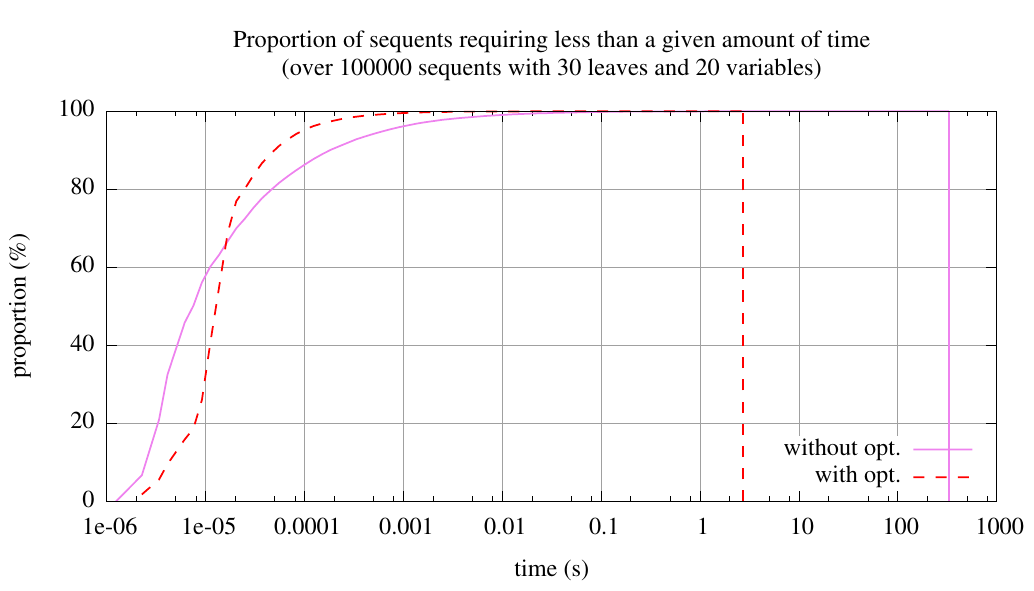}
  \caption{Distribution of searching times.}
  \label{fig:bench:distrib}
\end{figure}




\section{Conclusions and directions for future work}
\label{sec:ccl}

We proved untyping theorems for several standard structures, allowing
us to extend decidability results to the typed settings, and to
discover an optimisation of proof search for cyclic linear logic. All
results have been formally checked~\cite{this:web} with the Coq proof
assistant.

The untyping theorem for typed Kleene algebras is quite important in
the ATBR Coq library~\cite{atbr:itp}: it allows one to use our tactic
for Kleene algebras in typed settings, and, in particular, with
heterogeneous binary relations. The underlying decision procedure
being quite involved, we can hardly imagine proving its soundness with
respect to typed settings in a direct way. Even writing a
type-preserving version of the algorithm seems challenging.

At another level, we used the untyping theorem for semirings in order
to formalise Kozen's completeness proof~\cite{Koz94b} for Kleene
algebras. Indeed, this proof heavily relies on matrix constructions,
so that having adequate lemmas and tactics for working with possibly
rectangular matrices was a big plus: this allowed us to avoid the
ad-hoc constructions Kozen used to inject rectangular matrices into
square ones.

\subsection{References and related work}~\medskip
\label{ss:relwork}

The relationship between residuated lattices and substructural logics
is due to Ono and Komori~\cite{OnoK85};
see~\cite{galatos2007residuated} for a thorough introduction. Cyclic
linear logic was suggested by Girard and studied by
Yetter~\cite{Yetter90}. To the best of our knowledge, the idea of
adding types to the above structures is new. The axiomatisation of
Kleene algebras is due to Kozen~\cite{Koz94b}.

Our typed structures can be seen as very special cases of
\emph{partial} algebras~\cite{burmeister:partialalgebra}, where the
domain of partial operations is defined by typing
judgements. Similarly, one could use \emph{many-sorted}
algebras~\cite{higgins:manysorted} to mimic types using sorts.
Several encodings from partial algebras to total ones were proposed in
the literature~\cite{Mossakowski02,Diaconescu09}. Although they are
quite general, these results do not apply here: these encodings do not
preserve the considered theory since they need to introduce new
symbols and equations; as a consequence, ordinary untyped decision
procedures can no longer be used after the translation. Dojer has
shown that under some conditions, convergent term rewriting systems
for total algebras can be used to prove existence equations in partial
algebras~\cite{DBLP:journals/fuin/Dojer04}. While it seems applicable
to semirings, this approach does not scale to Kleene algebras or
residuated lattices, for which decidability does not arise from a term
rewriting system.

The idea of proving typed equations from untyped ones also appears in
the context of ``Pure Type Systems'' (PTSs), where one can use either
an untyped conversion rule, or a typed equality judgement. Whether
these two possible presentations were equivalent was open for some
time~\cite{GeuversWerner94}; Adams has shown that this is the case for
``functional'' PTSs~\cite{Adams06}, Herbelin and Siles recently
generalised the result to all PTSs~\cite{SilesHerbelin10}.  Although
the types we use here are quite basic (i.e., a type is just a pair of
abstract objects), our use of cut-free proof systems and factorisation
systems is reminiscent to their use of the Church-Rosser property.
Note however that unlike in functional programming languages, where
one usually relies on a Hindley-Milner type inference
algorithm~\cite{Hindley69,Milner78} to rule out ill-typed
programs, no inference algorithm is required with the algebraic theories
presented here: such an algorithm would always succeed since an
untyped proof systematically yields a typed proof.


Closer to our work is that of Kozen, who first proposed the idea of
untyping typed Kleene algebras, in order to avoid the aforementioned
matrix constructions~\cite{Koz98b}. He provided a different answer,
however: using model-theoretic arguments, he proved an untyping
theorem for the universal theory of ``$1$-free Kleene algebras''. The
restriction to $1$-free expressions is required, as shown by the
following counter-example: $\vdash 0=1 \Rightarrow a=b$ is a theorem
of semirings, although there are non trivial typed semirings where
$0=1$ holds at some types (e.g.,\ empty matrices), while $a=b$ is not
universally true at other types.


\subsection{Handling other structures}~\medskip
\label{ss:other}

\emph{Action algebras}~\cite{Pratt90,Jipsen04} are a natural extension
of the structures we studied in this paper: they combine the
ingredients from residuated lattices and Kleene algebras. In this
setting, left and right divisions make it possible to obtain a variety
rather than a quasi-variety: inference rules \rul{sl} and \rul{sr},
about the star operation, can be replaced by the following equational
axioms: %
\begin{mathpar}
  \inferrule*[Right=sl']{ }
  {\vdash \tstar{(a\backslash a)} = a\backslash a} \and
  \inferrule*[Right=sr']{ }{\vdash \tstar{(a/a)} = a/a} \and
\end{mathpar}
Although we do not know whether the untyping theorem holds in this
case, we can think of two strategies to tackle this problem: 1) find a
cut-free extension of the Gentzen proof system for residuated lattices
and adapt our current proof---such an extension is left as an open
question in~\cite{Jipsen04}, it would possibly entail decidability of
the equational theory of action algebras; 2) find a ``direct'' proof
of the untyping theorem for residuated monoids, without using a
Gentzen proof system, so that the methodology we used for Kleene
algebras can be extended.
Also note that we necessarily have to exclude the annihilator
element~$(0)$: with divisions, top~$(\top)$ can be defined as $0/0$,
so that the counter-example for bounded residuated lattices
(§\ref{ss:imall}) applies. Consistently, there is no way to remove
this element using a factorisation system: expressions like
$a\cdot{}\top$ cannot be simplified.

\emph{Kleene algebras with tests}~\cite{K97c} are another extension of
Kleene algebras, which is useful in program verification. Their
equational theory is decidable, but one cannot rely on a factorisation
system to remove annihilators in this setting: like for rings
(§\ref{ss:ring}), the complement operation of the Boolean algebra is
problematic. Moreover, like for Kleene algebras, there are no known
notions of normal form in Kleene algebras with tests, so that the
approach we described in~§\ref{ss:ring} is not possible. Nonetheless,
the untyping theorem is likely to hold for these structures since the
Boolean algebras of tests are inherently homogeneous.

Finally, although our methodology for semirings can be adapted to
handle the case of allegories [10] (see [33] for a proof), the cases
of distributive and division allegories---where left and right divisions
are added---remains open.

Finally, although our methodology for semirings can be adapted to
handle the case of \emph{allegories}~\cite{FreydScedrov90}
(see~\cite{this:web} for a proof), the cases of \emph{distributive
allegories} as well as \emph{division allegories}---where left and right
divisions are added---remain open.

\subsection{Towards a generic theory}~\medskip
\label{ss:flan}

The typed structures we focused on can be described in terms of
enriched categories, and the untyping theorems can be rephrased as
asserting the existence of faithful functors to one-object
categories. It would therefore be interesting to find out whether
category theory may help to define a reasonable class of structures
for which the untyping theorem holds. In particular, how could we
exclude the counter-example with additive constants in MALL?

For structures that are varieties, another approach would consist of
using term rewriting theory to obtain generic factorisation theorems
(Lemma~\ref{prp:s:factor}, which we used to handle the annihilating
element in semirings, would become a particular case). This seems
rather difficult, however, since these kind of properties are quite
sensitive to the whole set of operations and axioms that are
considered.

\section*{Acknowledgements.}
We are grateful to Olivier Laurent and Tom Hirschowitz for the highly
stimulating discussions we had on linear logic and about this work.

\bibliographystyle{plain}
\bibliography{bib}

\vspace{-40 pt}
\end{document}